\documentclass[12pt]{article}
\usepackage{graphicx}
\usepackage{amsmath,amsthm,amsfonts}
\usepackage{ifpdf}
\pdfoutput=1

\newtheorem{theorem}{Theorem}

\newtheorem{corollary}[theorem]{Corollary}

\def\inline#1:{\par\vskip 7pt\noindent{\bf #1:}\hskip 10pt}

\def\blackslug{\hbox{\hskip 1pt \vrule width 4pt height 8pt
    depth 1.5pt \hskip 1pt}}
\def\QED{\quad\blackslug\lower 8.5pt\null\par}

\def\Prob{\mathbb{P}}
\def\Exp{\mathbb{E}}

\def\policy{\mbox{\tt\bf P}}
\def\PP{\mbox{\tt P}}
\def\GWpolicy{\mbox{\bf\it p}}
\def\pp{p}
\def\actual{\mbox{\tt\bf A}}
\def\AA{\mbox{\tt A}}
\def\family{\mbox{\tt\bf F}}
\def\FF{\mbox{\tt F}}
\def\twochildplus{\mbox{\tt C}^{++}}
\def\onechild{\mbox{\tt 1C}}
\def\twochild{\mbox{\tt 2C}}
\def\ztwochild{\mbox{\tt 0/2C}}
\def\zthreechild{\mbox{\tt 0/3C}}

\def\onechildsub{\mbox{\tt\scriptsize 1C}}
\def\twochildsub{\mbox{\tt\scriptsize 2C}}
\def\zthreechildsub{\mbox{\tt\scriptsize 0/3C}}
\def\ztwochildsub{\mbox{\tt\scriptsize 0/2C}}

\def\chinasub{\mbox{\tt\scriptsize china}}
\def\indiasub{\mbox{\tt\scriptsize india}}
\def\SNone{\mbox{\tt SN}_{\onechildsub}}
\def\SNtwo{\mbox{\tt SN}_{\twochildsub}}
\def\SNztwo{\mbox{\tt SN}_{\ztwochildsub}}
\def\SNzthree{\mbox{\tt SN}_{\zthreechildsub}}
\def\SNchina{\mbox{\tt SN}_{\chinasub}}
\def\SNindia{\mbox{\tt SN}_{\indiasub}}

\begin{document}

\title{The Effect of Population Control Policies \\
on Societal Fragmentation
\thanks{Supported in part by the Israel Science Foundation (grant 1549/13).}
}

\author{Zvi Lotker
\thanks{Department of Communication Systems Engineering, Ben Gurion University of the Negev, Beer-Sheva, Israel. E-mails:~{\tt zvilo@cse.bgu.ac.il}}
\and
David Peleg
\thanks{Department of Computer Science, The Weizmann Institute of Science, 
Rehovot, Israel. E-mail:~{\tt david.peleg@weizmann.ac.il}}
}

\maketitle

\begin{abstract}
Population control policies are proposed and in some places employed
as a means towards curbing population growth.
This paper is concerned with a disturbing side-effect of such policies, 
namely, the potential risk of {\em societal fragmentation}
due to changes in the distribution of family sizes.
This effect is illustrated in some simple settings and demonstrated 
by simulation. In adition, the dependence of societal fragmentation
on family size distribution is analyzed. In particular, it is shown
that under the studied model, any population control policy
that disallows families of 3 or more children incurs the possible risk
of societal fragmentation.
\end{abstract}

\section{Introduction}

\subsection{Background}
Global population explosion is viewed by many as a major threat to the
well-being and stability of the human race.
Consequently, various population control policies were proposed,
with the goal of curbing the growth of the human population on earth.
Some of those policies were even implemented in different countries,
the Chinese one-child policy being the most well-known instance.

This paper is concerned with a less well-studied side-effect of such policies,
namely, the potential risk to the fabric structure of society, 
in the form of {\em societal fragmentation}.
Our main contribution is to demonstrate the potential occurrence of this
effect, by viewing society as a {\em social network} and employing
tools of social network theory.
Moreover, we demonstrate that different population control policies may yield
radically different outcomes in terms of fragmentation, even if they yield
essentially the same outcome in terms of the resulting population size.

A social network represents society as a graph, in which vertices represent
individual members of society and a link connecting two vertices
represents a social connection between the two individuals.
By ``societal fragmentation'' we refer to the situation where the social
network becomes disconnected, and breaks into a large number of separate 
(medium to large) connected components. In contrast, the network is considered
to be connected, or non-fragmented, if it consists of essentially one
large connected component (a so-called ``giant'' component), 
possibly along with some additional vanishingly small components.

Let us stress at the outset that a complete analysis of fragmentation 
as a function of population control policy in real networks requires 
a very precise and careful modeling, taking into account multiple parameters, 
and is outside the scope of the current paper.
Rather, our purpose is to illustrate the fragmentation effect and demonstrate 
its crucial dependence on 
the particular population control policies being utilized.
We therefore adopt a simplistic societal model, stripping away many of the
complicating parameters and focusing only on the pertinent features.

Societal connections are often classified into ``strong'' and ``weak'' ties,
see \cite{burt2009structural}. We employ a {\em strong-ties} network model 
of connectivity, namely, one that focuses only on strong ties, and ignores
weak ties.
Hence, we consider the social network as fragmented into separate 
connected components if there are no strong ties connecting those components, 
even if there exist some weak ties that connect them.

In particular, our model assumes that family ties are strong ties,
namely, there is a link connecting a person to his or her parents
and siblings, as well as between spouses.
In contrast, weaker types of social links 
(e.g., based on work/school relationships etc.) are ignored.

Viewing the recent Chinese history as a key example central to understanding
the phenomena of population reduction and societal fragmentation,
we seek to model social interconnections between the part of society
that was directly affected by the population control policy,
somewhat arbitrarily defined to be the generation born after 1980,
see \cite{China-1child-book-11,weibo-2015}.

Therefore, we view the social network as {\em layered} into generations,
and focus on analyzing the connectivity properties within a single generation
(particularly, the latest). In view of this, we ignore parent-child links
and include in our {\em strong-ties} network model only links between siblings
and spouses\footnote{We simplify our model by ignoring the fact that 
marital relations can be untied by divorce, in most of the world.}.
Hence in our model, a {\em strong-ties social network} is a social network 
where edges indicate sibling or marriage relations.

\subsection{Family parameters and fragmentation}

The size and fabric structure of society are largely determined by two central 
parameters, collectively referred to hereafter as the {\em family parameters},
namely, the {\em distribution of family sizes} and the {\em marriage ratio}.

The {\em distribution of family sizes} is described by a real vector
$$\family=(\FF_0,\FF_1,\FF_2,\ldots),$$
where $\FF_i\ge 0$ is the fraction of families with $i$ children, and 
$$\sum_i \FF_i=1.$$

The {\em marriage ratio} $\alpha$ is defined as the parentage of married couples
in the population. Clearly, there could be a difference between men and women.
Since we are interested in population control, and assuming monogamy,
we will measure the marriage ratio as the ratio between the total number
of married women and the total number of women.

In order to be able to analyze strong-ties social networks, 
we distinguish between the two different types of links based on family ties, 
namely, {\em sibling} links and {\em marital} links. In our graphical 
illustrations, these two types of links are drawn in {\em blue} and {\em red}, 
respectively.

As a first illustration of the effects of the family parameters on societal
fragmentation, let us consider two relevant examples, namely, the societal
structure of the populations of current-day China and India.

Based on the national statistics of China, see \cite{China2014},
the family size distribution $\family_{China}$ in urban society in China
(with the vector truncated after the first six positions, i.e.,
taking the last entry $\PP_5$ to represent the fraction of families 
with 5 or more children) is
$$\family_{China}=(0.418,0.269,0.17,0.085,0.039,0.019).$$
The marriage ratio 
in China is $\alpha_{China}=0.92$ \cite{NBSchina12}.

Figure \ref{fig:china} depicts a social network denoted $\SNchina$, 
generated by a simulation over a population of 157 individuals
based on the above family parameters.

\begin{figure}[htb]
\begin{center}
\includegraphics[scale=0.7]{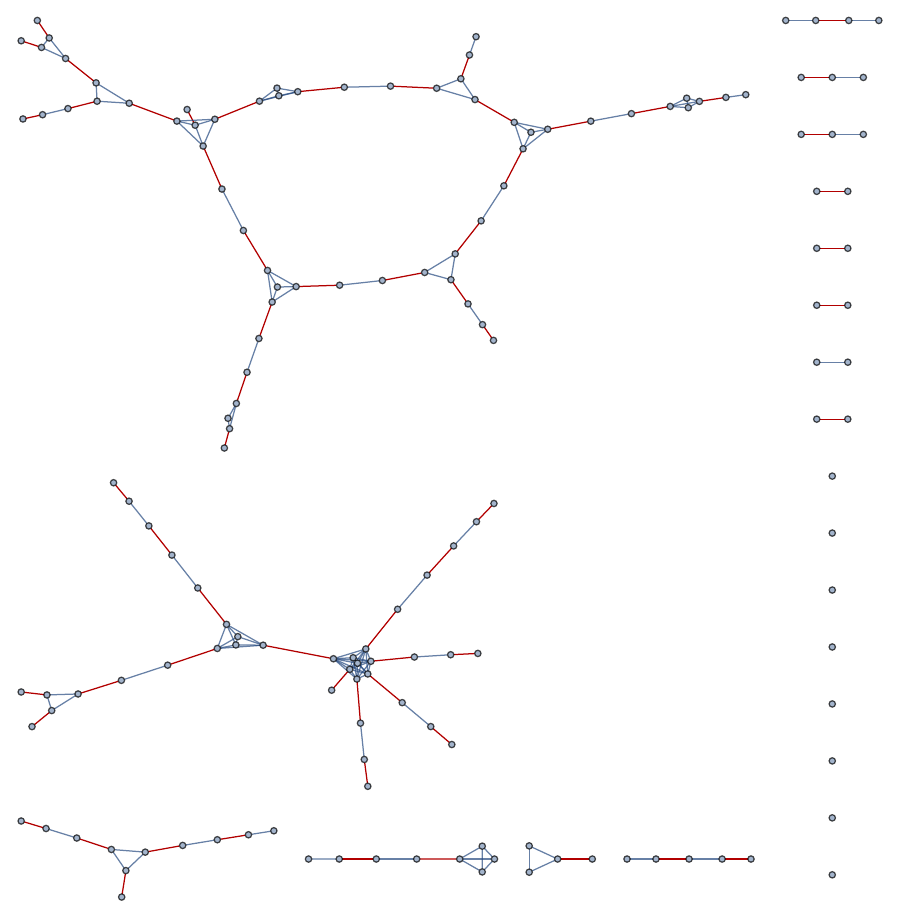}
\end{center}
\caption{\sf
A $157$ node network $\SNchina$, where the number of siblings simulates 
the distribution of siblings in large cities in China. 
Blue edges connect siblings, while red edges connect married couples. 
As can be seen, the network consists of a few large families, 
represented by the blue cliques in the graph. 
There are six relatively large connected components, 
and 16 smaller connected components, 8 singletons and 8 connected components, 
with 3 or fewer edges.
}
\label{fig:china}
\end{figure}

Similarly, based on recent data on family size in India, see \cite{India2011},
the family size distribution vector $\family_{India}$ in India 
(again truncated after the first six positions, so that $\PP_5$ represents 
the fraction of families with 5 or more children) is
$$\family_{India}=(0.126, 0.121, 0.199, 0.193, 0.141, 0.22).$$
Figure \ref{fig:india} depicts a social network denoted $\SNindia$, 
generated by a simulation over a population of 130 individuals
based on this family size distribution and an estimated\footnote{We were unable
to ascertain the exact data for India.} marriage ratio of $0.92$.

\begin{figure}[htb]
\begin{center}
\includegraphics[scale=0.75]{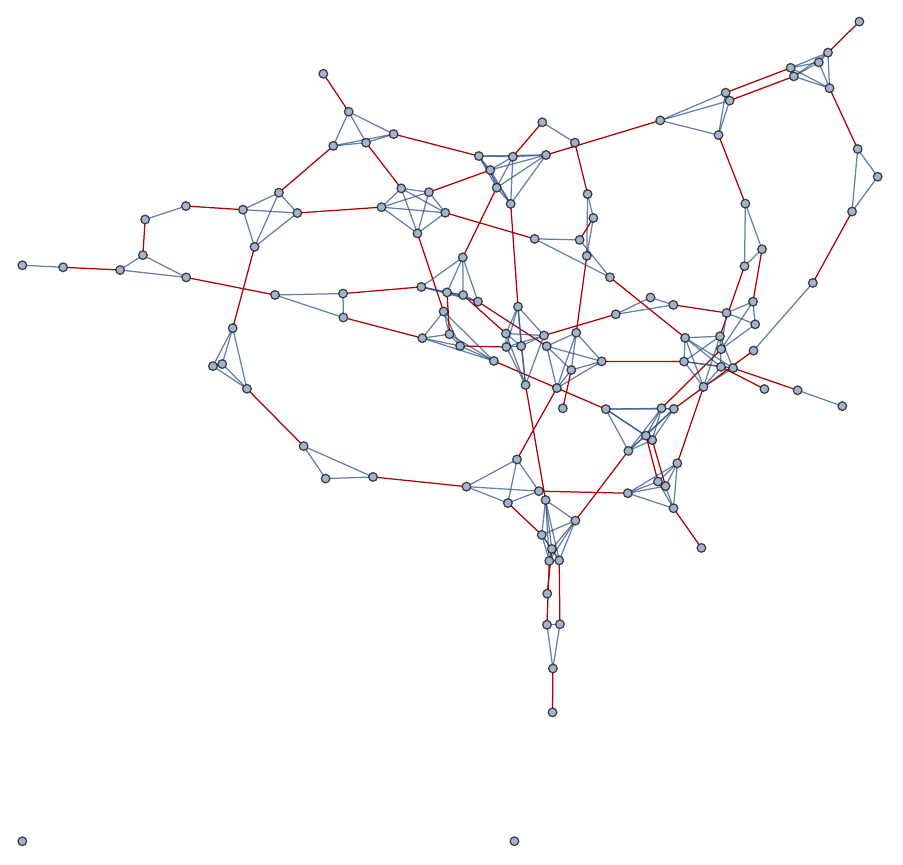}
\end{center}
\caption{\sf
A $130$ node network $\SNindia$, where the number of siblings simulates 
the distribution of siblings in India. The blue edges connect siblings. 
The average number of siblings is $2.8$, which is less than the real number 
in India, since the probability of having more than 7 siblings is assumed 
to be $0$. The red edges connect married couples. 
As can be seen, the network contains a single giant component.
}
\label{fig:india}
\end{figure}

As can be clearly seen, these two simulations exhibit strikingly different 
behavior between the resulting strong-ties social networks 
in terms of connectivity. Whereas the strong-ties social network $\SNchina$ 
is fragmented, the network $\SNindia$ remains connected and forms 
a giant component.

Our main goal in this paper is to try to explain and quantify the dependence 
of network fragmentation on the family parameters, in the hope of gaining 
a better understanding of the role and effects of population control policies.

\subsection{Results}

We begin (in Section \ref{sec:PCP}) by formally defining a class of 
population control policies and introducing a number of typical examples.
We then study the implications of these policies on societal fragmentation,
by simulating the implementation of these policies on a population of 
200 individuals, and inspecting the resulting strong-ties social networks.

Our simulations indicate that the family size distribution resulting from 
the employed population control policy directly affects fragmentation. 
More interestingly, it turns out that different policies may exhibit 
very different fragmentation effects, 
even when their effects on the population size are very similar.

We then turn to analyzing the probability that a given population control 
policy of the type described above will yield a fragmented society.
In order to carry out this analysis we turn to studying infinite populations.
This may appear odd, as clearly, population control deals with finite groups 
of people.
However, mathematical analysis of percolation systems tends to be easier
on infinite graphs. Moreover, there is a standard way to transform
the result from infinite graphs to finite graphs.
Basically, when dealing with infinite graphs the main question
in percolation theory is the appearance of an infinite size connected component.
In the context of a finite graph, this transforms into the existence of
a linear sized ``giant'' connected component.
Therefore, we prove our results on an infinite size population.

It follows from our results that 
strict policies, such as the 1-child policy employed in China until recently,
effectively reduces population size, but are likely to result 
in a highly fragmented society, composed of many separate components.
In fact, we prove that in order to avoid societal fragmentation, 
the population control policy must allow (at least a small fraction of) 
families of 3 or more children.

\section{Population control policies}
\label{sec:PCP}

In this section we introduce a simple class of population control policies, 
characterized as follows.
A {\em population control policy} is defined as a real vector 
$$\policy=(\PP_0,\PP_1,\PP_2,\ldots),$$
where $\PP_i\ge 0$ is the target fraction of families with $i$ children 
(or the desired probability that a random family will have $i$ children). 
These values should satisfy 
$$\sum_i \PP_i=1.$$

The condition imposed by the population control policy $\policy$ is that at 
any given time, the family size distribution $\family$
must satisfy the following condition:
\begin{equation}
\label{eq:act-vs-policy}
\hspace{0.5cm}
\mbox{For any }~ J\ge 0, 
\hspace{1.2cm}
\sum_{i=0}^J \FF_i ~\ge~ \sum_{i=0}^J \PP_i~.
\hspace{1.2cm}
\end{equation}
A possible method that can be applied in order to implement 
such a policy would be as follows. For every family $f$, draw a target integer 
$K(f)$ at random with the distribution $\policy$, 
i.e., setting $K(f)=i$ with probability $\PP_i$ for every $i$. 
Consequently, the maximum number of children allowed for the family $f$
will be $K(f)$.

In order to demonstrate the sensitivity of the strong-ties social network
to the population control policy used, 
we focus on a number of specific population control policies, 
to be described next.

We begin by comparing two basic policies, named hereafter 
the {\em 1-child policy} and the {\em $0/2$-children policy}.
Generally, these policies aim at reducing the population size by 50\% 
in each generation, by attempting to impose an average of one child per family.

\inline The 1-child policy ($\onechild$):
Under this policy, each family is allowed at most one child. 
This policy can be described by the vector 
$$\policy=(0,1,0,0,\ldots).$$
If imposed, such a policy ensures that at any given time, 
$$\FF_0+\FF_1 \ge 1$$
(or necessarily $\FF_0+\FF_1 = 1$),
i.e., there are no families with two or more children.

In the context of our discussion, this policy may be viewed as representing 
(a simplified version of) the original Chinese population control policy.

The main effect of the 1-child policy is that it strongly curbs 
the population size. 
For example, consider a generation consisting of a hundred men and a hundred 
women, where the marriage ratio is $\alpha=0.9$. The (approximately) 90
couples of this generation will have (at most) one child per family.
This means that the population size was reduced to well below 50\% of the size
of the previous generation.

What interests us about this policy, however, is the fact that it completely 
eliminates sibling ties from the network of the new generation.
Hence the network breaks down into about 45 (disconnected) couples
and a few unmarried individuals.
Figure \ref{fig:1c} depicts the resulting strong-ties social network $\SNone$
of the new generation under policy $\onechild$.

\begin{figure}[htb]
\begin{center}
\includegraphics[scale=0.5]{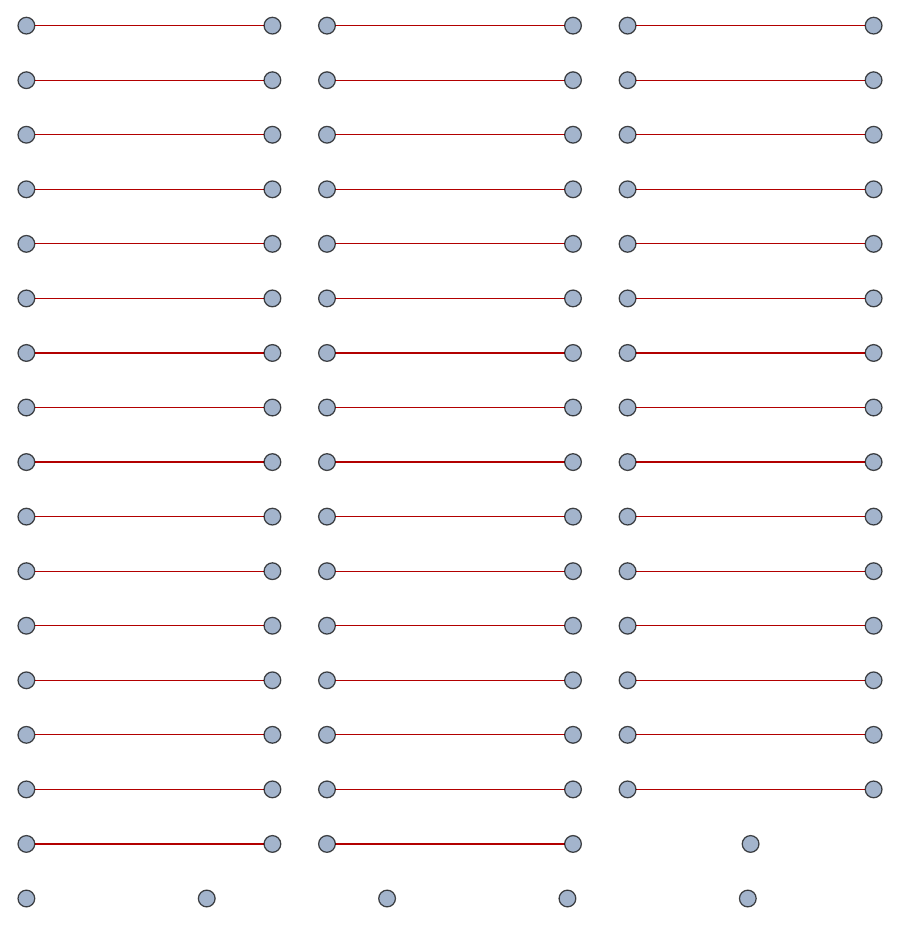}
\end{center}
\caption{\sf
The network $\SNone$ obtained by simulating the 1-child policy $\onechild$
starting with a generation of 200 individuals, with marriage ratio $\alpha=0.9$.
The network has no sibling edges (the only strong ties are marriage links),
hence it is totally fragmented into pairs. 
}
\label{fig:1c}
\end{figure}

\inline The $0/2$-children policy ($\ztwochild$):
This policy also attempts to maintain the average family size at one child.
It is described by the vector 
$$\policy = \left(\frac{1}{2},0,\frac{1}{2},0,0,\ldots\right),$$ 
namely, $\PP_0=\PP_2=1/2$ and all other probabilities are zero.
As discussed earlier, such a policy can be implemented by flipping 
an unbiased coin for each family and, depending on the outcome, 
allowing the family either zero or two children\footnote{We deliberately 
ignore societal, moral and philosophical issues, such as fairness,
involved in implementing such an ``arbitrarily heartless'' policy.}.

Observe that the outcome of applying the $0/2$-children policy is essentially 
almost identical to that of the 1-child policy  
in terms of controlling population size. 
On the other hand, we claim that it does a better job maintaining 
sibling ties and thus keeping society connected.

Indeed, consider again the scenario of a generation consisting of 
a hundred men and a hundred women, with a marriage ratio of $\alpha=0.9$.
Figure \ref{fig:02c} depicts the resulting strong-ties social network $\SNztwo$
of the next generation when applying policy $\ztwochild$ in this scenario.

Note that of the married couples of the first generation, about half will have 
(at most) two children per family, and the other half will have no children 
at all (hence these families will 
not be represented in the network of the next generation).
This means that the goal of reducing the population size by 50\% is achieved 
by the  0/2-child policy just as effectively as by the 1-child policy.
However, the 0/2-child policy produces a next generation that still has
sibling links, although perhaps not enough to maintain complete connectivity,
as implied by the simulated network appearing in Figure \ref{fig:02c}.

\begin{figure}[htb]
\begin{center}
\includegraphics[scale=0.5]{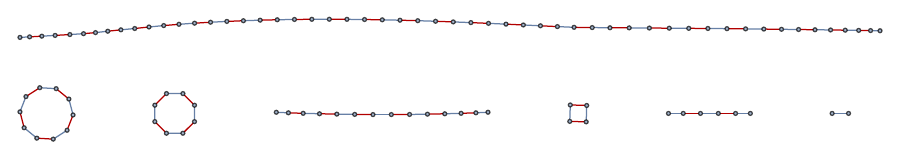}
\end{center}
\caption{\sf
The network $\SNztwo$ obtained by simulating the 0/2-child policy $\ztwochild$
starting with 200 individuals and $\alpha=0.9$.
}
\label{fig:02c}
\end{figure}

Next, we similarly compare two population control policies that aim at
maintaining the population size at steady state, i.e., attempt to impose
an average of two children per family.

\inline The 2-children policy ($\twochild$):
Under this policy, each family is allowed at most two children. 
This policy can be described by the vector 
$$\policy=(0,0,1,0,0,\ldots).$$
If imposed, such a policy ensures that at any given time, 
$$\FF_0+\FF_1+\FF_2=1,$$
i.e., there are no families with three or more children.

In the context of our discussion, this policy may represent 
(a simplified version of) the new Chinese population control policy.

Again, this policy keeps population size more or less stable, 
but it also limits the number of sibling ties in the network.
Returning to the scenario examined with the previous two policies,
of 200 individuals with marriage ratio $\alpha=0.9$,
under the 2-children policy the married couples of the first generation 
will have two children per family,
so the population size in the next generation will remain similar to 
(or slightly smaller than) that of the previous generation. 
Note, hhowever, that sibling ties will occur in the resulting network 
of the new generation with the same frequency as in the $\ztwochild$ policy. 

Figure \ref{fig:2c} depicts the strong-ties social network $\SNtwo$
of the new generation resulting in the above scenario under policy $\twochild$. 
Observe that this network is similar in structure to 
the network $\SNztwo$ of Figure \ref{fig:02c}, but is about twice its size. 

\begin{figure}[htb]
\begin{center}
\includegraphics[scale=0.5]{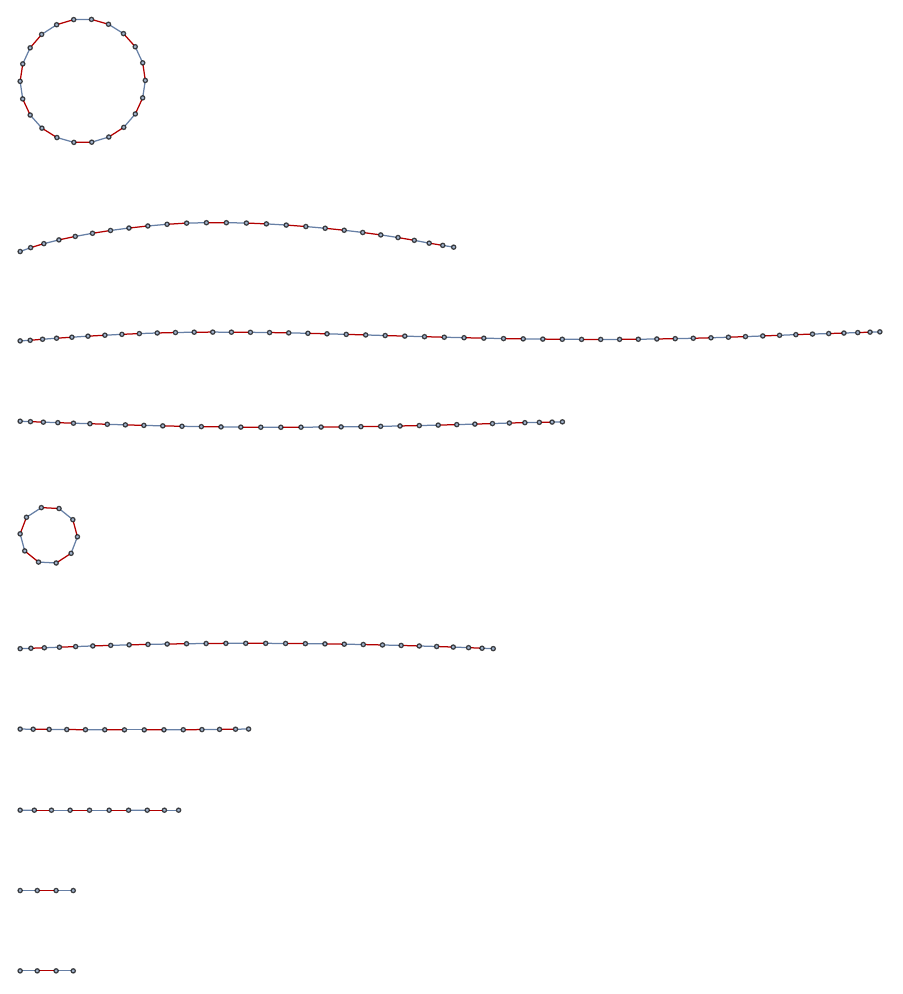}
\end{center}
\caption{\sf
The network $\SNtwo$ obtained by simulating the 2-child policy $\twochild$
starting with 200 individuals and $\alpha=0.9$.
}
\label{fig:2c}
\end{figure}

\inline The $0/3$-children policy ($\zthreechild$):
This policy also attempts to maintain the average family size at two children.
It is described by the vector 
$$\policy = \left(\frac{1}{3},0,0,\frac{2}{3},0,0,\ldots\right),$$ 
namely, $\PP_0=1/3$ and $\PP_3=2/3$ and all other probabilities are zero.
Such a policy can be implemented by flipping 
an biased $\frac{1}{3}~:~\frac{2}{3}$ coin for each family and, 
depending on the outcome, 
allowing the family either zero or three children.

The outcome of applying this policy is essentially almost identical to that 
of the 2-child policy in terms of controlling population size, but again, 
its impact on societal fragmentation is rather different.
This can be realized by inspecting Figure \ref{fig:03c}, 
which depicts the resulting strong-ties social network $\SNzthree$ 
of the new generation when applying policy $\zthreechild$
in the scenario discussed above. 

\begin{figure}[htb]
\begin{center}
\includegraphics[scale=0.5]{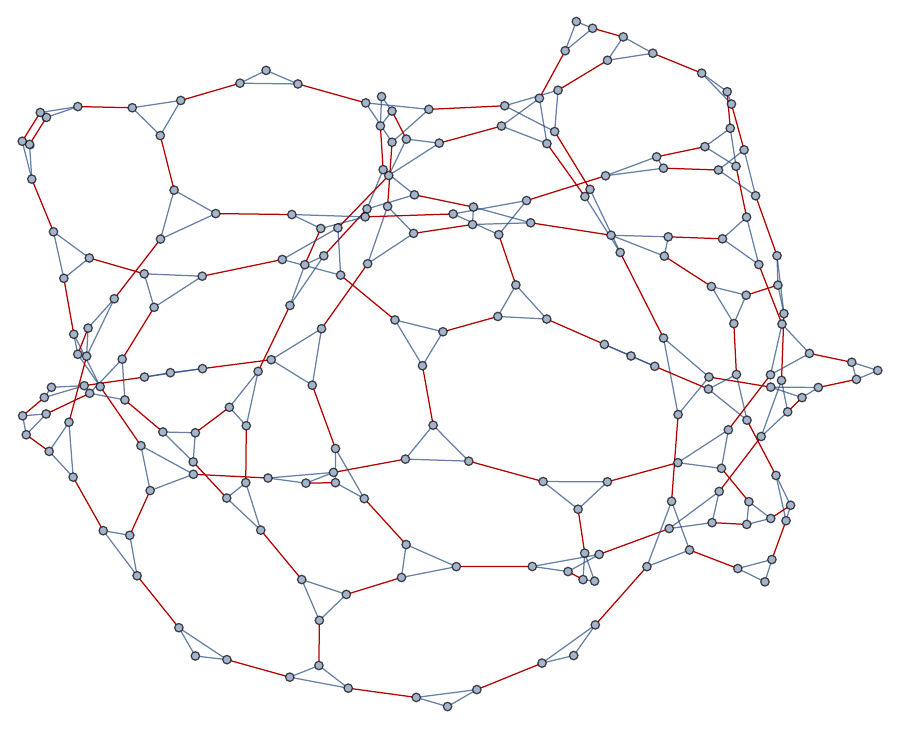}
\end{center}
\caption{\sf
The network $\SNzthree$ obtained by simulating the 0/3-child policy 
$\zthreechild$ starting with 200 individuals and $\alpha=0.9$.
}
\label{fig:03c}
\end{figure}

It should be realized that of the approximately 90 married couples 
of the original generation, about 30 families have no children,
hence these families are not represented in the network.
The other 60 families, however, have 3 children each, which yields 
a next generation with a population of about 180.
Most significantly, observe that in this network, the effect of the higher 
fraction of families with three children already suffices to ensure 
rather solid connectivity.

\section{Analysis}

In this section we analyze the likelihood that the strong-ties social network
becomes fragmented, and show that this is dependent on the two family 
parameters discussed above, namely, the family size distribution $\family$ 
and the marriage ratio $\alpha$, hence in turn it depends also on the 
population control policy $\policy$ that's being used.

One of the interesting implications of our analysis is that under the studied 
model, any population control policy that disallows families of 3 or more 
children incurs the possible risk of societal fragmentation.

Specifically, for the $0/3$-children population control policy $\zthreechild$, 
our results imply that there is a critical value for the marriage ratio 
$\alpha$, such that the network becomes fragmented for $\alpha$ values below 
this threshold, and is connected for higher $\alpha$ values.

As discussed earlier, we perform this analysis on an infinite size population,
and hence our definition for societal fragmentation changes with respect to 
the simulation based observations presented in the previous section.
We say that society is non-fragmented if an infinite size connected component 
emerges (with positive probability).

\subsection{The Galton-Watson branching process}

The {\em Galton-Watson branching process} (cf. \cite{Durrett2007}, pp. 31-32) 
generates a (potentially infinite) tree $T_{GW}$ rooted at $C_0$.
The {\em level} of a node in the tree is its distance from the root $C_0$
(whose level is 0). 

Let $\xi_i^{t}$, $i,t\geq 0$, be independent identically distributed
nonnegative integer valued random variables, 
where $\xi_i^t$ represents the number of children 
on level $t$ of the $i$th vertex on level $t-1$.
Let
$$\GWpolicy=(\pp_0,\pp_1,\ldots)$$
be a prescribed {\em child distribution}, where $p_k$ is the probability
for $k$ children,
$$p_k=\Prob[\xi_i^t=k].$$
Denote the expected value of $\xi_i^{t}$ by
$$\mu=\Exp[\xi_i^{t}].$$ 
The Galton-Watson process constructs the tree $T_{GW}$ iteratively 
as follows.
For the $i$th vertex on level $t-1$, the process assigns the variable 
$\xi_i^{t}$ a value at random according to the child distribution $\GWpolicy$,
and adds to this vertex exactly $\xi_i^{t}$ new children on level $t$.

Figure \ref{fig:GW} illustrates two possible outcomes of the 
Galton-Watson branching process.

\begin{figure}[htb]
\begin{center}
\includegraphics[scale=.3]{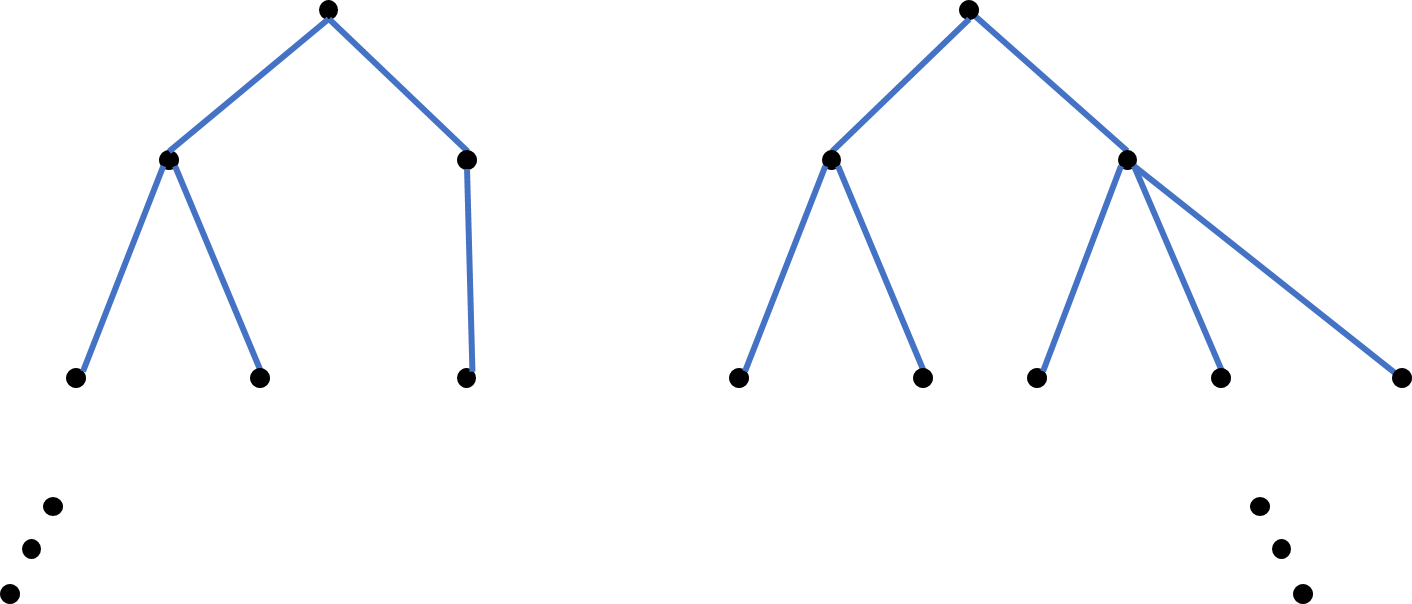}
\end{center}
\caption{\sf
Illustration of two possible outcomes $T_{GW}$ of the Galton-Watson 
branching process.
}
\label{fig:GW}
\end{figure}

Define the sequence $Z_{t}$, for all $t\geq 0$ by 
$$Z_0=1$$ 
and
\[ Z_{t+1} =
\left\{
  \begin{array}{ll}
    \xi_1^{t+1} +...+ \xi_{Z_{t}}^{t+1},  & \textrm{If } Z_{t}>0~,\\
    0, & \textrm{If } Z_{t}=0~.
  \end{array}
\right.
\]
The variable $Z_t$ represents the number of vertices at level $t$
(namely, at distance $t$ from the root).

We are now ready to cite Theorems 2.1.2, 2.1.3 and 2.1.4 
from \cite{Durrett2007}.
Let us denote by $P^{GW}_{\infty}$ the probability that an infinite size
connected component will emerge in the Galton-Watson branching process. 

\begin{theorem}
{\bf \cite{Durrett2007}}
\label{Thm:Durrett}
\begin{description}
\item[(a)]
If $\mu<1$ then $Z_{t}=0$ for all sufficiently large $t$.
(Hence $P^{GW}_{\infty}=0$, i.e., the resulting tree has no infinite size
connected component.)
\item[(b)]
If $\mu=1$ and $\Prob[\xi_i^{t}=1)<1$ then $Z_{t}=0$ 
for all sufficiently large $t$.
(Hence $P^{GW}_{\infty}=0$.)
\item[(c)]
If $\mu>1$ then $\Prob[Z_{t}>0]>0$ for all sufficiently large $t$.
(Hence $P^{GW}_{\infty}=0$, i.e., the tree 
has an infinite size connected component with nonzero probability.)
\end{description}
\end{theorem}

\subsection{The strong-ties branching process}

In order to analyze the situation where the probability of a family to have 
exactly $i$ children is limited by the population control policy $\policy$
to be at most $\FF_i$, we use the following {\em strong-ties branching process},
which is similar to the Galton-Watson process. 

The strong-ties branching process also constructs a (possibly infinite) tree 
$T_{ST}$ rooted at $C_0$. 
Here, however, each node in the tree represents a married couple.
The {\em level} of a node in the tree is its distance from the root $C_0$
(whose level is 0). 
The level of a {\em person} is the level of the node (i.e., the couple) 
it belongs to.

We start with a married couple $C_0$ serving as the root of the process. 
The root's children are {\em new nodes} representing the couples 
involving the married siblings
of (either one of the spouses in) the couple $C_0$. 
Note that each of those new nodes consists of two spouses, say, $a$ and $b$.
One of these spouses, say $a$, has exactly one married sibling in $C_0$
and all the other married siblings are on level $1$ in the tree. 
All the married siblings of the other spouse, $b$, are on level $2$.

This construction process proceeds by induction. 
Consider a node $C_t$, representid a married couple $a$ and $b$, 
on level $t$ of the tree generated by the branching process. 
Its children are new nodes representing the couples involving 
the married siblings
of $a$ and $b$ that are not in the tree already. 
For one of spouses, say $a$, one married sibling is on level $t-1$ 
and all the other are on level $t$.
(We hereafter refer to this spouse as the ``connected'' spouse.)
All the married siblings of the other spouse, $b$, are on level $t+1$ 
of the tree.
(We hereafter refer to this spouse as the ``new" spouse.)

\begin{figure}[htb]
\begin{center}
\includegraphics[scale=.3]{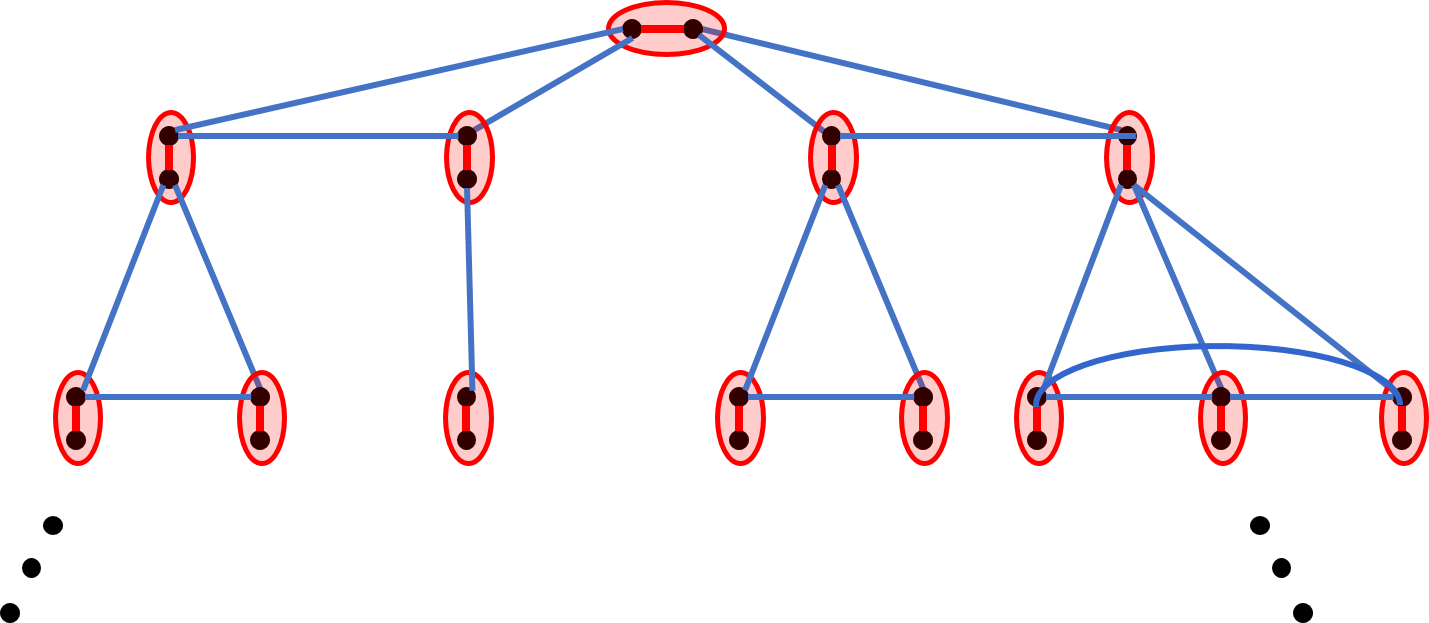}
\end{center}
\caption{\sf
A possible graph $G_{ST}$ constructed by the strong-ties branching process.
}
\label{fig:STa}
\end{figure}

Observe that if we take into account all the sibling links among the added 
nodes, then the constructed strong-ties graph $G_{ST}$ is not a tree,
since nodes corresponding to siblings are connected by a clique. Figure 
\ref{fig:STa} 
illustrates such a possible graph.

However, we may transform this graph into a strong-ties tree $T_{ST}$ by erasing
the clique connecting a group of siblings, and replacing it with a ``star'' 
connecting the first sibling (on some level $t$)
to all other sibling (on the next level $t+1$). Figure 
\ref{fig:STb} 
illustrates the strong-ties tree $T_{ST}$ corresponding to the strong-ties 
graph $G_{ST}$ of Figure 
\ref{fig:STa}.

\begin{figure}[htb]
\begin{center}
\includegraphics[scale=.3]{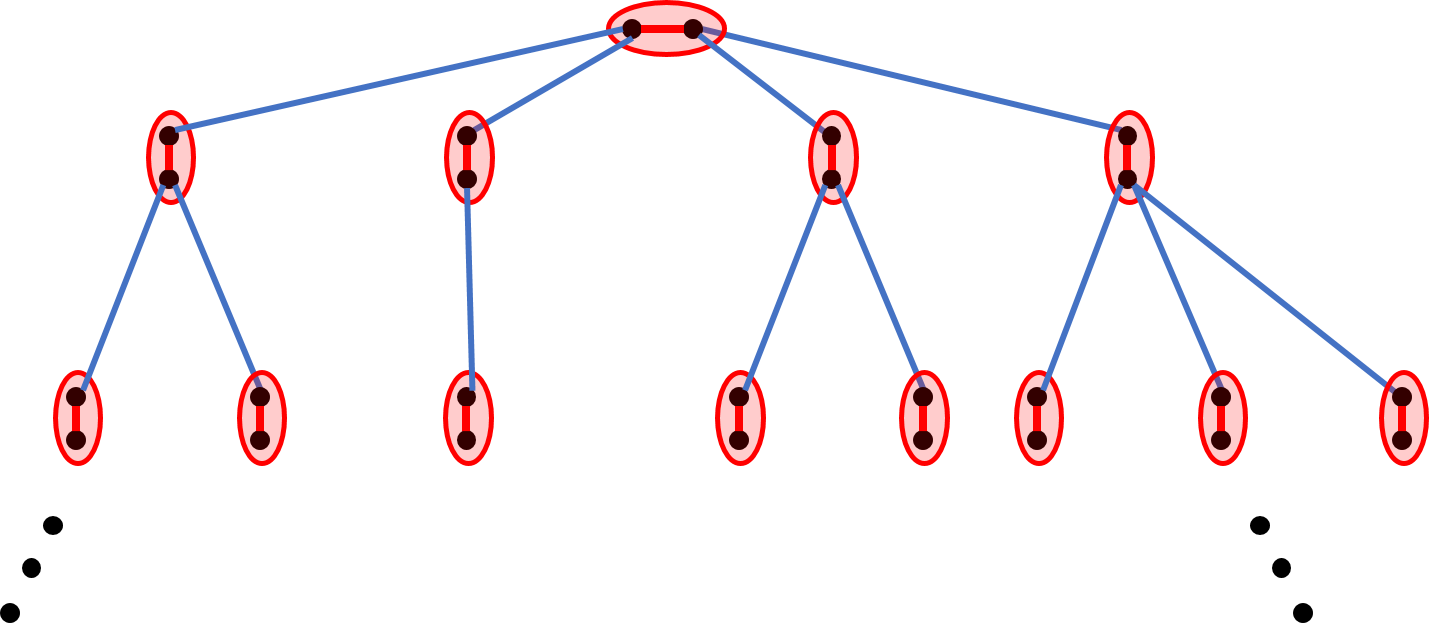}
\end{center}
\caption{\sf
The corresponding tree $T_{ST}$ obtained by replacing sibling cliques 
with stars.
}
\label{fig:STb}
\end{figure}

Note that the resulting strong-ties tree $T_{ST}$ is constructed 
rather similarly to the Galton-Watson tree $T_{GW}$.
The only technical difference is that in the strong-ties branching process,
the process is applied twice at the root $C_0$.
Hence, for example, the resulting tree $T_{ST}$ depicted in Figure
\ref{fig:STb}
corresponds to the tree obtained by
``fusing'' the roots of the two GW trees depicted in Figure \ref{fig:GW}.

To analyze the strong-ties branching process, we first compute the probability 
distribution of the actual number of children of an arbitrary non-root node. 
Denote this distribution by 
$$\actual=(\AA_0,\AA_1,\AA_2,\ldots),$$ 
where $\AA_i$ is the probability that a non-root node will actually have 
$i$ children in the tree. 
Note that the number of children of a node in our process 
depends on two factors: one is the number of sibling of the new spouses, 
and the other is how many of those are married. 
We next derive the formula for $\AA_i$, starting with the cases of 
$i=0/1$ children, and then analyzing the general case.

Consider first the case where a node $C$ on some level $t$ in the branching 
tree has $0$ children. This means that in the married couple of the node $C$, 
the new spouse comes from a family that has $k$ children for some $k\ge 1$, 
where only one of those is married and all other siblings are unmarried. 
Note that this married sibling is already in the tree (on level $t-1$). 
Therefore, the probability of this event to happen is
\begin{equation}
\label{eq:Galton-0}
\AA_0 ~=~ \sum_{k=1}^{\infty} \FF_k (1-\alpha)^{k-1}.
\end{equation}
Next, consider the case where a node $C$ on some level $t$ in the tree 
has one child. This means that in the married couple of the node $C$, 
the new spouse comes from a family that has $k$ children for some $k\ge 2$, 
where only two of those siblings are married and the rest are unmarried. 
The probability of this event to happen is
$$\AA_1 ~=~ \sum_{k=2}^{\infty} \FF_k {k-1 \choose 1} \alpha (1-\alpha)^{k-2}.$$
In the general case, where a node $C$ has $j$ children in the tree, 
the new spouse of $C$ comes from a family that havs $k\ge j+1$ children, 
where only $j+1$ of those siblings are married and the rest are not. 
The probability of this event to happen is
\begin{equation}
\label{eq:Galton-j}
\AA_j ~=~ \sum_{k=j+1}^{\infty} \FF_k {k-1\choose j}\alpha^j (1-\alpha)^{k-1-j}.
\end{equation}

We now use Theorem \ref{Thm:Durrett} to draw some conclusions on the size 
of the connected component of the strong-ties branching process 
that starts at $C_0$. 
(Note that, as explained before, the strong-ties tree $T_{ST}$ generated by 
the strong-ties branching process is similar to the Galton-Watson tree $T_{GW}$,
and therefore the theorem is applicable to it as well.)

According to Theorem \ref{Thm:Durrett}, the dynamics of the strong-ties 
branching process is completely determined by $\mu$. So to understand the 
emergence of an infinite size connected component
in the strong-ties social network, we need to determine the expected number 
of children in the strong-ties tree $T_{ST}$.
This, by our calculations, is equal to 
\begin{equation}
\label{eq:mu}
\mu ~=~ \sum_{j=1}^{\infty} j \cdot \AA_j.
\end{equation}
We thus get the following.
Denote by $P^{ST}_{\infty}$ the probability that an infinite size
connected component will emerge in the strong-ties branching process. 
\begin{theorem}
\label{thm:2}
Assume that $\alpha<1$, and consider any population control policy s.t. 
$\FF_i=0$ for all $i\ge 3$. Then $Z_{t}=0$ for all sufficiently large $t$
(hence $P^{ST}_{\infty}=0$).
\end{theorem}
\begin{proof}
By the definition of $\AA_j$, if $\FF_i=0$ for all $i\ge 3$, then by
Eq. (\ref{eq:Galton-j}) it follows that 
$$\AA_j=0 ~\mbox{ for all }~ j\geq 2.$$ 
This implies that 
$$\AA_0+\AA_1 = 1,$$ 
and in addition, by Eq. (\ref{eq:mu}), 
$$\mu=0\cdot \AA_0 + 1\cdot \AA_1 = \AA_1 = 1-\AA_0.$$ 
Note that since $\alpha<1$, Eq. (\ref{eq:Galton-0}) implies that $\AA_0>0$. 
It follows that $\mu<1$. The theorem now follows from Thm. \ref{Thm:Durrett}(a).
\end{proof}

\begin{corollary}
\label{cor:3}
For any population control policy $\policy$, having a nonzero fraction 
of families with 3 or more children (namely, $\PP_i>0$ for some $i\ge 3$) 
is a necessary condition for the emergence of an infinite size component 
in the strong-ties branching process.
\end{corollary}
\begin{proof}
By Theorem \ref{thm:2}, having $\FF_i>0$ for some $i\ge 3$ is a necessary 
condition for the emergence of an infinite size component in the strong-ties 
branching process.
Assume, towards contradiction, that in the population control policy $\policy$,
$\PP_i=0$ for every $i\ge 3$, and yet there exists some $i\ge 3$ for which
$\FF_i>0$. Then 
$$\PP_0+\PP_1+\PP_2=1$$ 
but 
$$\FF_0+\FF_1+\FF_2<1.$$ 
This contradicts Eq. (\ref{eq:act-vs-policy}) for $J=2$.
\end{proof}

\subsection{The strong-ties social network}

Observe that the strong-ties social network is not necessarily a tree.
The strong-ties branching process results in a tree $T_{ST}$ because
the nodes added to each level during the construction process 
are always {\em new} nodes. In contrast, in the strong-ties social network,
even assuming an infinite supply of new nodes, there's always the possibility
that a cycle will occur due to the spouse selection process.
In particular, if we try to mimic the branching process by looking at an 
arbitrary couple $C_0$, and constructing the network around it level by level 
as in the branching process,
it may happen that for a couple $C$ added to level $t$, 
some of the siblings of {\em both} spouses already appear on level $t-1$.

Denote by $P^{*}_{\infty}$ the probability that an infinite size
connected component will emerge in a strong-ties social network
on an infinite population. Then the above discussion implies that 
$$P^{*}_{\infty} \le P^{ST}_{\infty}.$$ 
Hence the conclusion of Cor. \ref{cor:3}
applies to strong-ties social networks as well.

Let us examine the two concrete examples of the 2-child policy $\twochild$
and the $0/3$-child policy $\zthreechild$ presented earlier.
Note that as a direct consequence of the last corollary, the 2-child policy
prevents the emergence of an infinite size connected component.
Hence the 2-child policy inherently yields a fragmented social network.

Turning to the $0/3$-child policy, the situation depends on the marriage rate
$\alpha$. For example, assume $\alpha=0.9$. Furthermore, assume that
each family utilizes its allocated number of children to the maximum.
Then 
$$\AA_j = \left\{
\begin{array}{ll}
\frac{2}{3}\cdot(1-\alpha)^2,  & j=0,
\vspace{10pt}
\\ 
\frac{2}{3}\cdot 2\cdot\alpha(1-\alpha),  & j=1,
\vspace{10pt}
\\
\frac{2}{3}\cdot\alpha^2,  & j=2,
\vspace{10pt}
\\
0,  &  j \ge 3.
\end{array}
\right.$$
Hence 
$$\mu=2\cdot \frac{2}{3} \cdot\alpha = 1.2,$$ 
implying that such a policy stands a chance to keep society connected 
and result in an infinite size connected component
(dependent on other parameters omitted from the discussion here, e.g., 
the possibility that a family allowed 3 children will actually have 
fewer children).

\section{Conclusion}

In this paper we examine the role of population control policies
and show that they affect (simultaneously) two independent issues: 
they manage global population growth, but also affect 
the connectivity of strong social network ties. 

We show that in order to avoid fragmentation of the strong-ties 
social network, it should be allowed to have families with three or more 
children, in order to overcome the (however small) percentage of 
unmarried people, which is an unavoidable phenomenon. 

The paper gives several examples for population control policies,
and shows that while some of them provide an equal guarantee on 
the number of people, there are vast differences in terms of the resulting
fragmentation in the strong-ties social network.

A somewhat dissatisfying aspect of our model of population control policies 
is that some of the policies introduced in our model 
(such as $\ztwochild$ or $\zthreechild$) involve a certain degree
of unfairness, as they assign non-uniform family sizes. What's worse, 
the assignment is done in an arbitrary way. This makes these policies 
inherently problematic from a moral point of view.
To make such policies more viable, it may be desirable to consider ways
to reduce or restrict their level of arbitrariness, and perhaps augment them 
by introducing some social mechanisms of exception handling.

To illustrate these considerations, and possible approaches towards 
handling them, let us give yet another concrete example for a possibly
more attractive population control policy.
Consider a policy such as the following, which, for lack of a better name, 
we might call the 
$\twochildplus$ policy. 
Its probability distribution is
$$\policy=\left(0,0,\frac{9}{10},\frac{1}{10},0,0,\ldots\right).$$
It can be implemented by allowing each family two children, 
and in addition selecting 10\% of the families and granting them
permission for a third child. This selection may be done randomly
(thus ensuring at least some minimal degree of fairness), or alternatively 
by applying a variety of social and economic criteria.

Such a policy (or one designed along similar principles) may succeed 
in curbing population size, keeping in mind the fact that the marriage ratio 
is below 1. Specifically, assuming $\alpha=0.92$, the expected population
reduction in one generation when using policy $\twochildplus$ is 3.4\%, 
even assuming each family utilizes its quota of allowed children in full.

At the same time, the $\twochildplus$ policy succeeds also in preventing 
societal fragmentation. Note that the expected value of $\xi_i^{t}$ 
when using policy $\twochildplus$ is 
$$\mu = 1.1 \cdot \alpha,$$
so assuming specifically $\alpha=0.92$, we get $\mu = 1.012,$
hence by Theorem \ref{Thm:Durrett}(c), an infinite size connected component 
will emerge (in an infinite population) with nonzero probability.

Finally, the $\twochildplus$ policy can be seen 
as significantly more balanced, and less harsh, compared to the ones 
discussed earlier.



\begin{thebibliography}{1}

\bibitem{burt2009structural}
Ronald~S Burt.
\newblock {\em Structural holes: The social structure of competition}.
\newblock Harvard university press, 2009.

\bibitem{Durrett2007}
Rick Durrett.
\newblock {\em Random Graph Dynamics}.
\newblock Cambridge Univ. Press, 2007.

\bibitem{China-1child-book-11}
Esther C.~L. Goh.
\newblock {\em China's One-Child Policy and Multiple Caregiving: Raising little
  suns in Xiamen}.
\newblock Routledge, 2011.

\bibitem{weibo-2015}
Wentao Han, Xiaowei Zhu, Ziyan Zhu, Wenguang Chen, Weimin Zheng, and Jianguo
  Lu.
\newblock Weibo, and a tale of two worlds.
\newblock In {\em Proc. 2015 IEEE/ACM Int. Conf. on Advances in Social Networks
  Analysis and Mining (ASONAM)}, pages 121--128, 2015.

\bibitem{India2011}
Census of~India.
\newblock {HH-01 Normal Households By Household Size}, 2011.

\bibitem{NBSchina12}
National~Bureau of~Statistics~of China.
\newblock {\em women and Men in China: Facts and Figures 2012}.
\newblock 2012.

\bibitem{China2014}
National~Bureau of~Statistics~of China.
\newblock {Table 2-15 Family households by size and region (2013)}, 2014.

\end{thebibliography}
\end{document}